\documentclass[aps,pra,showpacs,twocolumn,nofootinbib,10pt]{revtex4-1}
\pdfoutput=1

\usepackage{graphicx}
\usepackage{amsmath}
\usepackage{amssymb}
\usepackage{mathrsfs}
\usepackage{amsthm}
\usepackage{bm}
\usepackage{url}
\usepackage[T1]{fontenc}
\usepackage{csquotes}
\MakeOuterQuote{"}
\newtheoremstyle{note}
  {\topsep}               % ABOVE SPACE
  {\topsep}               % BELOW SPACE
  {}                      % BODY FONT
  {\parindent}            % INDENT (empty value is the same as 0pt)
  {\itshape}              % HEAD FONT
  {.}                     % HEAD PUNCTUATION
  {5pt plus 1pt minus 1pt}% HEAD SPACE
  {}

\theoremstyle{note}
\newtheorem{theorem}{Theorem}
\newtheorem{lemma}{Lemma}
\newtheorem{conjecture}{Conjecture}

\theoremstyle{definition}

\theoremstyle{remark}
\newtheorem{remark}{Remark}

\newcommand{\mrm}[1]{\mathrm{#1}}
\providecommand{\tr}{\operatorname{tr}}

\newcommand{\bbF}{\mathbb{F}}

\newcommand{\be}{\begin{equation}}
\newcommand{\ee}{\end{equation}}
\newcommand{\ba}{\begin{align}}
\newcommand{\ea}{\end{align}}

\def\<{\langle}  %% overiding the original command \<
\def\>{\rangle}  %% overiding the original command \>

\newcommand{\GL}[2]{\mrm{GL}(#1,#2)}

\newcommand{\AGL}[2]{\mrm{AGL}(#1,#2)}

\def\eqref#1{\textup{(\ref{#1})}}  %% overiding the original command \eqref
\newcommand{\eref}[1]{Eq.~\textup{(\ref{#1})}}

\newcommand{\thref}[1]{Theorem~\ref{#1}}
\newcommand{\Thref}[1]{Theorem~\ref{#1}}

\newcommand{\lref}[1]{Lemma~\ref{#1}}

\newcommand{\cref}[1]{Conjecture~\ref{#1}}
\newcommand{\Cref}[1]{Conjecture~\ref{#1}}

\newcommand{\rcite}[1]{Ref.~\cite{#1}}
\newcommand{\rscite}[1]{Refs.~\cite{#1}}

\begin{document}
	\title{Nonexistence of  sharply covariant mutually unbiased bases in odd prime dimensions}
	\author{Huangjun Zhu}
	\email{hzhu@pitp.ca}
	\affiliation{Perimeter Institute for Theoretical Physics, Waterloo, Ontario N2L 2Y5, Canada}

%03.67.-a quantum information
%03.65.Wj quantum tomography, state reconstruction
%02.10.De algebraic structure
%03.65.-w quantum mechanics

%03.67.Mn Entanglement production, characterization and manipulation
%03.65.Ud Entanglement and quantum nonlocality
%(e.g. EPR paradox, Bell's inequalities, GHZ states, etc.)
%(for entanglement production in quantum information, see 03.67.Mn);

\begin{abstract}

Mutually unbiased bases (MUB) are useful in a number of research areas. The symmetry of MUB is an elusive and interesting subject.
A (complete set of) MUB in dimension $d$ is sharply covariant if it can be generated by a group of order $d(d+1)$ from a basis state. Such MUB, if they exist, would be most appealing to theoretical studies and practical applications. Unfortunately, they seem to be quite rare. Here we  prove that no MUB in odd prime dimensions is sharply covariant, by virtue of clever applications of Mersenne primes, Galois fields, and Frobenius groups. This conclusion provides valuable insight about the symmetry of MUB and the geometry of quantum state space. It complements and strengthens the earlier result of the author that only two stabilizer MUB are sharply covariant. Our study leads to the conjecture that no MUB other than those in dimensions 2 and 4 is sharply covariant.

\end{abstract}
\date{\today}
\maketitle

In a $d$-dimensional Hilbert space,
two bases $\{|\psi_j\rangle\}$ and
$\{|\phi_k\rangle\}$ are \emph{mutually unbiased} if all the transition
probabilities $|\langle \psi_j|\phi_k\rangle|^2$ across their basis states are equal to $1/d$ \cite{Ivan81,WootF89,DurtEBZ10}.
Since their discovery, mutually unbiased bases (MUB) have found numerous applications, such as in the determination of quantum states, in the study of quantum kinematics, and in the construction of generalized Bell states (see  \rcite{DurtEBZ10} for a review).
Any  MUB in dimension $d$ contains at most  $d+1$ bases; the MUB is complete if the upper bound is attained.
In the rest of the paper by an MUB we shall mean  such a complete set.
When $d$ is a prime power,
 MUB can be constructed explicitly
\cite{Ivan81,WootF89,DurtEBZ10}, and most MUB can be constructed from stabilizer states \cite{Gott97the,BandBRV02,LawrBZ02, GodsR09, Kant12, Zhu15Sh, Zhu15M}.
All known constructions rely on the existence of Galois fields, which admit no generalization to any other dimension.

An MUB  in dimension $d$ is  \emph{group  covariant} if it can be generated from a single state---the  \emph{fiducial state}---by a group composed of unitary transformations.  The MUB is   \emph{sharply covariant} if in addition the generating group (modulo phase factors) can be chosen to have the minimum possible order of $d(d+1)$. Sharply covariant MUB, if they exist, would be most appealing from both theoretical and practical points of  view.  The basis states in such an MUB can be labeled naturally by group elements, and all of them can be generated from a single basis state by applying transformations in the minimal generating group without repetition.
Sharp covariance is highly desirable not only for MUB, but also for many other elusive discrete symmetric structures, such as discrete Wigner functions \cite{Woot87, GibbHW04, Zhu15P}, symmetric informationally complete measurements (SICs) \cite{Zaun11, ReneBSC04,ScotG10, Zhu12the, ApplFZ15G}, and other quasi-probability representations of quantum mechanics \cite{FerrE08,FerrE09}.

Although many MUB are group covariant,  it turns out that sharply covariant MUB are quite rare.
Among all  stabilizer MUB, only two (those in dimensions 2 and 4) are sharply covariant up to unitary equivalence, according to a recent result of  the author \cite{Zhu15Sh}.
Is there any other MUB that is sharply covariant if we go beyond stabilizer MUB? This question is of intrinsic interest to understanding the geometry of quantum state space and may have potential implications for foundational studies. However, it is very difficult to answer this question in general since most existing literature on this subject has focused on stabilizer  MUB \cite{BandBRV02,LawrBZ02, GodsR09,DurtEBZ10, Kant12, Zhu15Sh, Zhu15M}, in particular, the one that is equivalent to the MUB originally constructed in \rscite{Ivan81,WootF89}. Little is known about MUB beyond stabilizer construction.

In this paper we prove that no MUB in odd prime dimensions is sharply covariant without any additional assumption.  This conclusion represents one of only a few results that apply to all MUB in odd prime dimensions.  It complements and strengthens the earlier result of the author concerning stabilizer MUB \cite{Zhu15Sh}. Our study  leads to the conjecture that no MUB other than those in dimensions 2 and 4 is sharply covariant. This conjecture, if it holds, would reshape our understanding about the characteristics  of MUB.
The techniques introduced here may set the stage  for future research  on the symmetry of  MUB and related subjects.
Our study also
sharpens the  contrast between MUB  and SICs, all known examples of which are sharply covariant \cite{Zaun11, ReneBSC04,ScotG10, Zhu12the, ApplFZ15G}. In addition, all group covariant SICs in prime dimensions are covariant with respect to Heisenberg-Weyl groups and are thus sharply covariant~\cite{Zhu10}. Our work is of intrinsic interest to understanding the geometry of quantum state space.

In the course of our study, we shall reveal  surprising
connections  between MUB  and a number of other subjects, such as
Mersenne primes \cite{CaldH09book}, Galois fields \cite{Ash07book}, and Frobenius groups \cite{Feit67book,KurzS04book}. All these subjects are of broad interest to researchers from diverse fields.
Although Galois fields feature prominently in the study of MUB, their appearance in the current context exhibits  quite a different nature.

\begin{theorem}\label{thm:MUBcov}
	No MUB in any odd prime dimension is sharply covariant.
\end{theorem}
\begin{remark}
The MUB in dimension 2 is sharply covariant with respect to any order-6 subgroup in its symmetry group \cite{Zhu15Sh}. By symmetry group  we mean the group of all unitary transformations modulo phase factors that leave the MUB invariant.
\end{remark}
Before proving \thref{thm:MUBcov}, we need to introduce several technical tools.  A positive-operator-valued measure (POVM) $\{\Pi_j\}$ is a set of positive operators that sum up to the identity.
A POVM  $\{\Pi_j\}$ is informationally complete (IC) if all states can be distinguished by the  probabilities of measurement outcomes or, equivalently, if the outcomes  $\Pi_j$ span the operator space. The POVM is covariant with respect to a given group $G$ (composed of unitary operators) if it can be generated from one of the outcomes, say $\Pi_1$, under the action of $G$. Since operators that differ only by overall phase factors implement the same transformation, it is convenient to identify such operators. The group obtained from $G$ after such identification is called the collineation group of $G$ and denoted by $\overline{G}$. The following lemma establishes a simple connection between informationally completeness of a group covariant POVM and the irreducibility of the generating group; see \rcite{Zhu15Sh} for a proof.
\begin{lemma}\label{lem:ICirr}
Suppose $\{\Pi_j\}$ is an IC POVM that is covariant with respect to the group $\overline{G}$. Then $\overline{G}$ is irreducible.
\end{lemma}
\begin{remark}
	This lemma still applies even if $\overline{G}$ contains antiunitary transformations.
\end{remark}

Since the POVM constructed from all the  basis states of an MUB is IC, any group that can generate the MUB is necessarily irreducible according to \lref{lem:ICirr}. If an MUB in prime dimension $p$ is sharply covariant with respect to a group, then the generating group (modulo phase factors) has order $p(p+1)$ and thus defines a faithful irreducible projective representation of a group of order $p(p+1)$. To achieve our goal, therefore, it is instructive to look into those groups of order $p(p+1)$ that admit  a faithful irreducible projective representation of degree $p$. The following theorem shows that there is at most one such  group up to isomorphism.

\begin{theorem}\label{thm:MGFS}
	Let  $G$ be a group of order $p(p+1)$ with $p$ an odd prime. Suppose $G$ has a faithful irreducible projective  representation of degree $p$. Then $p$ is a Mersenne prime and  $G$ is isomorphic to the one-dimensional affine general linear group $\AGL{1}{p+1}$  over the Galois field with $p+1$ elements.
\end{theorem}

\begin{remark}
	A \emph{Mersenne prime} $p$ is a prime of the form $2^n-1$ with $n$ a positive integer; note that $n$ is necessarily also a prime \cite{CaldH09book}.
\end{remark}
Before proving this theorem, we need to introduce a technical lemma.
\begin{lemma}\label{lem:SylowNum}
Let  $G$ be a group of order $p(p+1)$ with $p$ a prime. Suppose $G$ has an  irreducible projective representation of degree $p$. Then $G$ has $p+1$ Sylow $p$-subgroups.
\end{lemma}
\begin{proof}
Recall that any group of order 6 is either cyclic or isomorphic to the symmetric group on three letters. Since any irreducible projective representation of a cyclic group has degree 1, it follows that $G$  is  isomorphic to the symmetric group on three letters when $p=2$,  from which the lemma follows immediately.

 When $p$ is an odd prime, suppose  on the contrary that $G$ does not have $p+1$ Sylow $p$-subgroups. Then it has only one Sylow  $p$-subgroup according to Sylow's theorem~\cite{KurzS04book}. Let $P$ be the Sylow $p$-subgroup of $G$ and $C$ its centralizer, then both $P$ and $C$ are  normal in $G$.
 In addition, the quotient group $G/C$ can be identified with a subgroup of the automorphism group of $P$. Observing that the automorphism group of $P$ is cyclic of order $p-1$~\cite{KurzS04book}, whose greatest common divisor with $p(p+1)$ is 2, we conclude that $C$ has index either 1 or 2 in $G$.
	
Let   $\varphi$ be an irreducible projective representation of $G$ of degree $p$ and $h$  a generator of $P$. If $C=G$, then $h$ commutes with all elements in $G$, so that  $\varphi(h)$ commutes with all elements in $\varphi(G)$ up to phase factors. Observing that any element of $G$ not in $P$ has order coprime with $p$, we conclude that $\varphi(h)$ commutes with all elements in $\varphi(G)$. Therefore, $\varphi(h)$ is proportional to the identity according to Schur's lemma. Consequently, the representation $\varphi$ induces an irreducible projective  representation of $G/P$ of degree $p$. This is impossible, however, because the degree of any irreducible projective representation of a finite group is bounded from above by the square root of the order of the group, note that the image of the representation necessarily spans the operator space (cf.  Corollary 2.30 in \rcite{Isaa76}). Alternatively, we may reach contradiction by observing that any irreducible projective  representation of $G$ of degree $p$ must be faithful.

If $C$ is a subgroup of $G$ of index 2, then $C$ is normal in $G$. Therefore, the restriction of $\varphi$ on $C$ has at most two irreducible components, and all the components have the same degree. Given that $p$ is an odd prime, the restriction must be irreducible. However,
any irreducible projective representation of a group of order $p(p+1)/2$ has  degree at most $\sqrt{p(p+1)/2}<p$. This contradiction confirms the claim that $G$ has $p+1$ Sylow $p$-subgroups.
\end{proof}

Another stepping stone for proving \thref{thm:MGFS} is the famous Frobenius theorem \cite{Feit67book, KurzS04book}\footnote{One reason the Frobenius theorem is famous is because no purely group theoretic proof has been found for this purely group theoretic result despite numerous attempts of many researchers for more than a century.
All known proofs  rely on some representation theory.}.
A nontrivial proper subgroup $H$  of the group $G$ is a \emph{Frobenius complement} of $G$ if $H\cap H^g=1$ for all $g\in G\setminus H$, where $H^g=g^{-1} H g$ and 1 denotes the identity of $G$ as well as the group composed of the identity. The group $G$ is called a \emph{Frobenius group} if it possesses such a complement.  Frobenius theorem says that the \emph{Frobenius kernel}
\begin{equation}
K:=(G\setminus \cup_{g\in G}H^g)\cup{1}
\end{equation}
is a normal subgroup of $G$. In addition,  the Frobenius complements of a Frobenius group are unique up to conjugation, so the Frobenius kernel is unique.

\begin{proof}[Proof of \thref{thm:MGFS}]
	According to \lref{lem:SylowNum} and Sylow's theorem, $G$ has  $p+1$
     Sylow $p$-subgroups, all of which are conjugate to each other.  Let $P$ be a Sylow $p$-subgroup of $G$; then  the normalizer of $P$ is itself, so  the intersection $P\cap P^g$ is trivial for all $g\in G\setminus P$. Therefore, $G$ is a Frobenius group with Frobenius complement $P$. According to the Frobenius theorem \cite{Feit67book, KurzS04book}, the Frobenius kernel $K$ of $G$ is a normal subgroup of order $p+1$.

     Note that any nontrivial element of $K$ cannot commute with any nontrivial element of $P$. It follows that all nontrivial elements of $K$ are conjugate to each other under $G$ and thus have the same order, which is necessarily a prime. Consequently,  $K$ has order  a prime power and thus has a nontrivial center (cf. Sec.~3.1.11 in \rcite{KurzS04book}), which in turn implies that
     $K$ is an elementary Abelian group. Suppose $K$ has order $r^n=p+1$, where  $r$ is a prime and $n$  a positive integer. Then  $n>1$ and $r=2$ since, otherwise, $r^n-1$ cannot be an odd prime. Therefore, $p$ is  a prime of the form $2^n-1$, that is, a  Mersenne prime.
	
	According to the above analysis, $G$ is a semidirect product of $P$ and $K$, where $K$ can be identified as an $n$-dimensional vector space over the Galois field $\bbF_2$. So $P$ can be identified as a subgroup of the general linear group $\GL{n}{2}$ of order $p=2^n-1$, that is,  a Singer cyclic subgroup \cite{Sing38, Bere00, Zhu15Sh, Zhu15P}.  Observing that all Singer cyclic subgroups of  $\GL{n}{2}$ are conjugate to each other and that the affine general linear group $\AGL{1}{p+1}$ (taken as a subgroup of $\AGL{n}{2}$) is a semidirect product of a Singer cyclic subgroup and $K$, we conclude that $G$ is isomorphic to $\AGL{1}{p+1}$.
\end{proof}

\begin{lemma}\label{lem:AffinePermu}
	Let  $K$ be the Sylow $r$-subgroup of $\AGL{1}{r^n}$, where $r$ is a prime and $n$  a positive integer. Let  $\varphi$ be a faithful irreducible projective representation of  $\AGL{1}{r^n}$ of degree $r^n-1$.  Then all operators in $\varphi(K)$ commute with each other and all   common eigenspaces of $\varphi(K)$ are nondegenerate. In addition, $|\tr(\varphi(g))|^2=1$ for all $g\in K^*$, where $K^*$ is the set of  elements in $K$ other than the identity. Any  order-($r^n-1$) element  of $\AGL{1}{r^n}$ induces  a cyclic permutation over  the basis states in the common eigenbasis of $\varphi(K)$\footnote{According to the answer of Derek Holt to a question of the author posted on MathOverflow, the Schur multiplier of the group $\AGL{1}{r^n}$ is trivial except when $r^n=4$ (http://mathoverflow.net/questions/191885/what-is-the-schur-multiplier-of-the-affine-linear-group-agln-q). Therefore, any  irreducible projective representation of $\AGL{1}{r^n}$ with $r^n\neq 4$ can be turned into an ordinary representation, although this fact is not necessary in the proof of \lref{lem:AffinePermu}.}.
\end{lemma}
\begin{proof}
Since $K$ is an elementary Abelian $r$-group, every irreducible projective representation has degree a power of $r$. Observing that $K$ is a normal subgroup of $G=\AGL{1}{r^n}$, we conclude that all the irreducible components of the restriction of $\varphi$ on $K$ have the same degree, which is a divisor of $r^n-1$. Given that  $r^n-1$ is not divisible by $r$, it follows that all the components have degree 1. The sum of squared multiplicities of inequivalent irreducible components of $\varphi(K)$ satisfy
\begin{equation}
	r^n-1\leq \frac{1}{r^n} \sum_{g\in K} |\tr(\varphi(g))|^2\leq \frac{1}{r^n}\sum_{g\in G} |\tr(\varphi(g))|^2=r^n-1.
\end{equation}
Therefore,
\begin{equation}
	 \frac{1}{r^n}\sum_{g\in K} |\tr(\varphi(g))|^2=r^n-1, \quad  \sum_{g\in K^*} |\tr(\varphi(g))|^2=r^n-1,
\end{equation}
which implies that all the irreducible components have the same multiplicity 1, so that  the common eigenbasis of $\varphi(K)$ is well defined. In addition,
$|\tr(\varphi(g))|^2=1$ given that all elements of $K^*$ are conjugate to each other in $G$.

The common eigenbasis of $\varphi(K)$ is invariant under the action of $\varphi(G)$  because $K$ is normal in $G$. On the other hand, $G$ acts transitively on the basis states since the representation is irreducible. Observing that $G$ is generated by $K$ and any order-($r^n-1$) element,  we conclude that any  order-($r^n-1$) element induces  a cyclic permutation over  the basis states in the common eigenbasis of $\varphi(K)$.
\end{proof}

\begin{lemma}\label{lem:MUBcov2n1}
        No MUB in dimension $2^n-1$ with $n\geq2$ is sharply covariant with respect to any collineation group that is isomorphic to
        	 $\AGL{1}{2^n}$.
\end{lemma}

\begin{proof}
	Suppose on the contrary that there exists an MUB in  dimension $2^n-1$ that is sharply covariant with respect to a collineation group that is isomorphic to  $G:=\AGL{1}{2^n}$. Then this collineation group is necessarily irreducible according to \lref{lem:ICirr} and thus defines  a faithful irreducible projective representation $\varphi$ of $G$ of degree $2^n-1$.  Let $K$ be the Sylow $2$-subgroup of $G$, then $K$ is an elementary Abelian $2$-group of order $2^n$. According to \lref{lem:AffinePermu}, we can choose a suitable basis such that  all operators in $\varphi(K)$ are diagonalized simultaneously. In addition, $|\tr(\varphi(g))|^2=1$ for $g\in K^*$, so $\varphi(g)$  have $2^{n-1}-1$ eigenvalues equal to 1 and $2^{n-1}$ eigenvalues equal to $-1$ after choosing suitable phase factors.

Suppose $|\psi\rangle$ is a fiducial ket of the MUB with respect to the common eigenbasis of $\varphi(K)$. Then $|\langle \psi| \varphi(g)|\psi\rangle|^2=1/(2^n-1)$ for each $g\in K^*$ since $g$ cannot stabilize any basis. Given that $\varphi(g)$ is real diagonal, it follows that
\begin{equation}\label{eq:MUBcon}
\langle \psi| \varphi(g)|\psi\rangle=\pm \frac{1}{\sqrt{2^n-1}}.
\end{equation}
Observing that $\sum_{g\in K^*} \varphi(g)$ commutes with all operators in $\varphi(G)$, we conclude that $\sum_{g\in K^*} \varphi(g)$ is proportional to the identity and is thus equal to $-1$ given that $\sum_{g\in K^*} \tr(\varphi(g))=-(2^n-1)$.
	Summing over $g\in K^*$ in \eref{eq:MUBcon} yields $-1=m/\sqrt{2^n-1}$, where  $-(2^n-1)\leq m\leq 2^n-1$ is an integer. However, this equality can never hold for $n\geq2$ because the left hand side is rational, while the right hand side is irrational, note that  $2^n-1$ is not a perfect square unless $n=1$. This contradiction completes the proof of the lemma.
\end{proof}

\begin{proof}[Proof of \thref{thm:MUBcov}]
	Suppose on the contrary that there exists an MUB  in odd prime dimension $p$ that is sharply covariant with respect to the collineation group $\overline{G}$. Then $\overline{G}$  has order $p(p+1)$ and is irreducible according to \lref{lem:ICirr}; in other words, it defines a faithful irreducible projective representation of a group of order $p(p+1)$. According to \thref{thm:MGFS}, $p$ is a Mersenne prime, that is, $p=2^n-1$ for some integer $n\geq 2$, and $\overline{G}$ is isomorphic to  $\AGL{1}{2^n}$.   However, no MUB can be sharply covariant with respect to such a group $\overline{G}$ according to \lref{lem:MUBcov2n1}. This contradiction confirms the claim of the theorem that no MUB in any odd prime dimension is sharply covariant.
\end{proof}

\Thref{thm:MUBcov} still holds even if antiunitary transformations are taken into account.
\begin{theorem}\label{thm:MUBstabCov2}
	No  MUB in odd prime   dimensions  is sharply covariant with respect to any group composed of unitary or antiunitary transformations.
\end{theorem}
\begin{proof}
	Suppose on the contrary that there is an MUB in odd prime dimension $p$ that is sharply covariant with respect to $\overline{G}$.
	Then $\overline{G}$ has order $p(p+1)$ and is irreducible according to \lref{lem:ICirr}. In addition, $\overline{G}$ must contain some antiunitary transformations according to \thref{thm:MUBcov}.
	Let $\overline{H}$ be the subgroup of $\overline{G}$ that is composed of unitary transformations. Then $\overline{H}$ has index 2 and is normal  in $\overline{G}$.  Consequently,  $\overline{H}$ is either irreducible or has two irreducible components of equal degree. The second  possibility cannot happen given that  the dimension $p$ is an odd prime. The first possibility cannot happen either because any irreducible collineation group in dimension $p$ has order at least $p^2$, while $\overline{H}$ has order $p(p+1)/2$. This contradiction completes the proof of the theorem.
\end{proof}

In summary, we have proved that no MUB in odd prime dimensions is sharply covariant, and the conclusion still holds even if antiunitary transformations are taken into account.
Our study in this paper and \rcite{Zhu15Sh} have accumulated fairly strong evidence that sharply covariant MUB are quite rare. It is plausible that such MUB can  appear only in dimensions 2 and~4.
\begin{conjecture}
	No MUB is sharply covariant except for those in dimensions 2 and 4.
\end{conjecture}
Resolution of this conjecture would provide valuable insight about the structure of MUB and the geometry of the quantum state space.
It would also sharpen
the  contrast between two elusive discrete symmetric structures, namely, MUB and SICs.
Recall that all known SICs are sharply covariant, and it is plausible that this conclusion also holds for all SICs yet to be discovered. The deep reason behind this contrast is still not clear. The potential implications of these observations for foundational studies deserve further exploration.  We hope our work will stimulate further interest and progress along this direction.

\section*{Acknowledgments}The author is grateful to Derek Holt for answering a question on MathOverflow concerning the Schur multiplier of the affine general linear group over a finite field.
 This work is supported in part by Perimeter Institute for Theoretical Physics. Research at Perimeter Institute is supported by the Government of Canada through Industry Canada and by the Province of Ontario through the Ministry of Research and Innovation.

\bigskip
\bibliographystyle{apsrev4-1}

\bibliography{all_references}

\end{document}